\documentclass[12pt]{article}
\usepackage{amsmath,amsfonts,amsthm,filecontents}

\usepackage[a4paper]{geometry}
\newtheorem{theorem}{Theorem}

\newcommand{\R}{\mathbb R}
\newcommand{\Q}{\mathbb Q}

\newcommand{\Z}{\mathbb Z}
\newcommand{\ee}{\varepsilon}
\newcommand{\SSS}{\mathcal S}

\begin{document}
\sloppypar

\title{Improved Approximation Algorithms for the Min-Max Selecting Items Problem}
\author{Benjamin Doerr\\
Max Planck Institute for Computer Science,\\
Campus E1 4\\
66123 Saarbr\"ucken, Germany}

\newcommand\eps{\varepsilon}
\maketitle

\begin{abstract}
  We give a simple deterministic $O(\log K / \log\log K)$ approximation algorithm for the Min-Max Selecting Items problem, where $K$ is the number of scenarios. While our main goal is simplicity, this result also improves over the previous best approximation ratio of $O(\log K)$ due to Kasperski, Kurpisz, and Zieli\'nski (Information Processing Letters (2013)). Despite using the method of pessimistic estimators, the algorithm has a polynomial runtime also in the RAM model of computation. We also show that the LP formulation for this problem by Kasperski and Zieli\'nski (Annals of Operations Research (2009)), which is the basis for the previous work and ours, has an integrality gap of at least $\Omega(\log K / \log\log K)$. 
\end{abstract}

\section{A Simple Approximation Algorithm}

In this short note, we first give a simple approximation algorithm for the Min-Max Selecting Items problem. In this problem, we are given $n$ items numbered from $1$ to $n$ and a set $\SSS$ of $K$ scenarios. A scenario $S \in \SSS$ is an assignment of nonnegative integral costs $c_{S,i}$ to each item $i \in [n] := \{1, \ldots, n\}$. The goal is to select a set $P$ of exactly $p$ items such that the maximum cost $c(P,S) := \sum_{i \in P} c_{S,i}$ of the selected items in any scenario is minimal. This problem belongs to the class of problems where solutions are sought which are robust to an event determined only after the optimization process. See~\cite{KouvelisY97} for more on robust optimization.

For the above defined Min-Max Selecting Items problem, a randomized $O(\log K)$ approximation algorithm was given in~\cite{KasperskiZ09}. A deterministic algorithm of same approximation ratio was given in~\cite{KasperskiKZ13}. We shall now give an algorithm considerably simpler than the two previous ones that achieves an approximation ratio of $O(\log K / \log\log K)$.

Like both previous works, we regard the following class of linear programs. For any $C \ge 0$, let $I_C := \{i \in [n] \mid \forall S \in \SSS : c_{S,i} \le C\}$ denote the set of items having cost at most $C$ in all scenarios. Consider the linear program
\begin{align*}
LP_C: &\sum_{i \in I_C} x_i = p\\
			&\sum_{i \in I_C} c_{S,i} x_i \le C, S \in \SSS\\
			&x_i \in [0,1], i \in I_C.
\end{align*}

Via binary search over the reasonable values for $C$ and solving $LP_C$, we find the smallest $C$ such that $LP_C$ has a solution. Naturally, this $C$ is a lower bound for the optimum of the Min-Max Selecting Items problem. 

Our aim in the following is to transform the fractional solution $(x_i)_{i \in I_C}$ into a solution for the selecting items problem that has cost $c(P,S) = O(C \log K / \log\log K)$ for all $S \in \SSS$.

The solution we will construct shall only take items from $I_C$, hence for convenience, we may simply assume $I_C = [n]$. By scaling the costs, we may assume that $C = 1$, and consequently (since $I_C = [n]$), that all costs are rational numbers not exceeding $1$. 

We now use dependent randomized rounding and its derandomization to find the desired solution. In Raghavan's~\cite{Raghavan88} classic randomized rounding, we would round each $x_i$ independently to some 0,1 valued $y_i$ such that $\Pr(y_i = 1) = x_i$. Chernoff bounds would immediately give that with reasonable probability, we have $\sum_{i \in [n]} c_{S,i} y_i \le O(\log K / \log\log K)$ for all $S \in \SSS$. The problem with this approach is that the cardinality constraint $\sum_{i \in [n]} y_i = p$ is unlikely to be satisfied. Our feeling is that overcoming this difficulty is the major reason why the previous solutions to the Min-Max Selecting Items problem are slightly technical. 

Fortunately, there is an easy solution. As first shown in 2001 by Srinivasan~\cite{Srinivasan01}, one can do randomized rounding both satisfying cardinality constraints and satisfying the same Chernoff bounds that are known for independent randomized rounding. This idea has found numerous applications in the last ten years. On the technical side, an alternative solution for this rounding problem that led to the first derandomization of such roundings was developed in~\cite{ichstacs05}. That Srinivasan's approach can be derandomized, in fact by simply reusing Raghavan's pessimistic estimators, was shown in~\cite[Theorem~3.1]{DoerrW09}. We state a slightly different formulation of this result here, which has an identical proof.
\begin{theorem}\label{trr}
  Let $A \in [0,1]^{m \times n}$ and $x \in [0,1]^n$. Assume that $\sum_{i \in [n]} x_i \in \Z$. Then in time $O(mn)$, a $y \in \{0,1\}^n$ can be computed such that $\sum_{i \in [n]} y_i = \sum_{i \in [n]} x_i$ and such that the rounding errors $|(Ay)_r - (Ax)_r|$ in each row $r$ satisfy the same bounds as in Raghavan's derandomization of independent randomized rounding. This result assumes that one can compute rational powers of integers with perfect precision in constant time. In the RAM model of computation, the result only holds for $A \in \{0,1\}^{m \times n}$.
\end{theorem}

We apply this result to our solution $x$ of $LP_C$ and the matrix $A := (C_{S,i})$. By adding dummy variables, we may assume that $s_r := (Ax)_r = 1$ for all rows $r$. Then the theorem above together with Theorem~3 and equation~(1.14) in~\cite{Raghavan88} gives a $y \in \{0,1\}^n$ with the rounding errors in all rows bounded by, in Raghavan's notation, $\Delta(1,1/(2K)) \le \frac{e \ln(2K)}{\ln(e \ln(2K))}$. Consequently, $P = \{i \in [n] \mid y_i = 1\}$ is a solution to the Min-Max Selecting Items problem with maximum cost at most $1 + \frac{e \ln(2K)}{\ln(e \ln(2K))} = O(\log K / \log\log K)$.

\section{Making the Algorithm Work in the RAM Model of Computation}

One known restriction of Raghavan's derandomization (described in detail in Section~2.2 of the original paper~\cite{Raghavan88}) is that it can be implemented in the RAM model of computation only if the coefficients of the constraint matrix (this is the matrix $A$ in Theorem~\ref{trr}) are in $\{0,1\}$. In all other cases, rational powers of integers have to be computed, so the derandomization can only be executed in the Real RAM model. 

Since having to assume the Real RAM model of computation for a purely combinatorial problem is undesirable (see also~\cite{SrivastavS96}), there has been some interest to make the method of pessimistic estimators also work in the RAM model. The first successful derandomization for the RAM model was given by Srivastav and Stangier~\cite{SrivastavS96}. As the 30-page paper indicates, their approach is technically quite involved. A second price to pay is an increased runtime of $O(mn^2 \log(mn))$ as opposed to the usual $O(mn)$ runtime for Raghavan's derandomization. 

An alternative approach was presented in~\cite{ichstacs06}, which solves many derandomization problems in time $O(mn \log n)$ by a reduction to Raghavan's solution for $\{0,1\}$ matrices. Since the derandomization result as formulated in~\cite{ichstacs06} does not give $O(\log K / \log\log K)$ approximations (in our notation), we quickly prove an alternative formulation that serves our needs.

\begin{theorem}
  Let $A \in ([0,1] \cap \Q)^{m \times n}$ and $x \in ([0,1] \cap \Q)^n$. Assume that $\sum_{i \in [n]} x_i \in \Z$. Then in time $O(mn \log n)$, a $y \in \{0,1\}^n$ can be computed such that (i) $\sum_{i \in [n]} y_i = \sum_{i \in [n]} x_i$ and (ii) for all $r \in [m]$, $(Ay)_r = O(\max\{1,(Ax)_r\} \log m / \log\log m)$.
\end{theorem}

\begin{proof}
  Let $\ell = \lceil\log_2 n\rceil$. In time $O(mn\log n)$, compute binary matrices $A^{(1)}, \ldots, A^{(\ell)} \in \{0,1\}^{m \times n}$ such that $\tilde A := \sum_{j \in [\ell]} 2^{-j} A^{(j)}$ and $A$ differ in each entry by at most $2^{-\ell} \le 1/n$. By the triangle inequality, $\|Ax - \tilde Ax\|_\infty \le 1$ for all $x \in [0,1]^n$.
  
 Applying Theorem~\ref{trr} to $x$ and the $(m\ell) \times n$ matrix obtained from all rows of $A^{(1)}, \ldots, A^{(\ell)}$, in time $O(m\ell n)$ in the RAM model we obtain a $y \in \{0,1\}^n$ such that (i) $\sum_{i \in [n]} y_i = \sum_{i \in [n]} x_i$ and (ii') for all $r \in [m]$ and $j \in [\ell]$, we have that $|(A^{(j)} y)_r - (A^{(j)} x)_r|$ satisfies the upper bound for rounding error  obtained by Raghavan's derandomization for an $(m\ell) \times n$ matrix. Note that by (1.13), (1.14) and Theorem~3 in~\cite{Raghavan88}, each of these rounding errors is $O(\max\{1,(A^{(j)}x)_r\} \log(m\ell) / \log\log(m\ell))$. Consequently,
\begin{align*}
  |(\tilde Ay)_r - (\tilde Ax)_r| 
    &= \bigg|\sum_{j \in [\ell]} 2^{-j} ((A^{(j)}y)_r - (A^{(j)}x)_r) \bigg|\\
    &\le \sum_{j \in [\ell]} 2^{-j} |(A^{(j)}y)_r - ((A^{(j)}x)_r| \\
    &\le O(\log(m\ell) / \log\log(m\ell)) \sum_{j \in [\ell]} 2^{-j} \max\{1,(A^{(j)}x)_r\}   \\
    &\le O(\log(m\ell) / \log\log(m\ell)) \sum_{j \in [\ell]} 2^{-j} (1 + (A^{(j)}x)_r)   \\
    &\le O(\log(m\ell) / \log\log(m\ell))(1 + (\tilde A x)_r),
\end{align*}
and thus $(Ay)_r \le 1 + (\tilde A y)_r \le 1 + (\tilde Ax)_r + |(\tilde Ay)_r - (\tilde Ax)_r| = O(\max\{1,(\tilde Ax)_r\} \log(m\ell) / \log\log(m\ell)) = O(\max\{1,(Ax)_r\} \log(m\ell) / \log\log(m\ell))$. This shows the theorem for, e.g., $m \ge \sqrt{\log n}$.

  If $m < \sqrt{\log n}$, we may use elementary linear
algebra as follows to transform $x$ into a vector $x' \in ([0,1] \cap \Q)^n$ such that $Ax = Ax'$, the cardinality constraint $\sum_{i \in [n]} x'_i = \sum_{i \in [n]} x_i$ is satisfied, and at most
$m+1$ entries of $x'$ are not $0$ or $1$: Let $J \subseteq [n]$, $|J| = m+2$ such that $x_j \notin \{0,1\}$ for all $j \in J$. Then by essentially solving an $m \times (m+2)$ system of linear equalities (in time $O(m^3)$), we obtain an $\ee \in \R^n \setminus \{0\}$ such that $\ee_{|[n]\setminus J} = 0$, $A\ee = 0$ and $\sum_{i \in [n]} \ee_i = 0$. Hence adding a suitable multiple of $\ee$ to $x$ yields a $[0,1]$--vector $x'$ having fewer non-integral entries than $x$ and still satisfying $Ax = Ax'$ and the cardinality constraint. Repeating this $O(n)$ times, we end up with the desired  $x'$. Computing it took $O(nm^3) \le O(nm\log n)$ time. 

We can now ignore the entries of $x'$ that are already $0$ or $1$ and the corresponding columns of $A$. We solve the resulting derandomization problem consisting of an $m \times (m+1)$ matrix $\tilde A$, a cardinality constraint, and an $(m+1)$--dimensional vector $\tilde x$ by simply checking all at most $2^{m+1} = O(n)$ possible roundings and computing their rounding errors each in time of order $m^2 = O(m \log n)$. Again, the total time for this is $O(mn\log n)$. This procedure finds a $y$ as desired, since we know its existence 
from Theorem~\ref{trr} already.
\end{proof}

\section{The Integrality Gap}

We now show that the linear relaxation $LP_C$ has an integrality gap of at least $\log K / \log\log K$, that is, there is an instance of the Min-Max Selecting Items problem such that $LP_1$ is feasible, but any integral solution to this Min-Max Selecting Items problem has cost at least $\Omega(\log K / \log\log K)$. This indicates that LP-based approaches using this LP formulation will not easily give  approximation ratios asymptotically better than $\Theta(\log K / \log\log K)$.

Let $k$ be an arbitrary integer. Let $p \ge k$ and $n \ge k^2 + (p-k)$. For each $T \in \binom{[k^2]}{k} := \{T \subseteq [k^2] \mid |T| = k\}$, define a scenario $S_T$ by $c_{S_T,i} = 1$, if $i \in T$, $c_{S_T,i} = 0$, if $i \in [k^2 + (p-k)] \setminus T$, and $c_{S_T,i} = 2$ otherwise. Let $x \in [0,1]^n$ be defined by $x_i = 1/k$ for $i \in [k^2]$, $x_i = 1$ for $i \in [k^2+1..k^2+(p-k)]$, and $x_i = 0$ otherwise. Then $\sum_{i \in [n]} x_i = p$ and $\sum_{i \in [n]} c_{S_T,i} x_i = 1$ for all $T \in \binom{[k^2]}{k}$. Hence $LP_1$ is feasible. 

Now let $P$ be an optimal (integral) solution to this problem instance. Since items in $[k^2+(p-k)+1..n]$ have cost 2 in all scenarios, whereas those in $[k^2+(p-k)]$ have cost at most 1 in all scenarios (and these are at least $p = |P|$ items), the optimality of $P$ implies $P \subseteq [k^2 + (p-k)]$. Since $|P| = p$, we have $|P \cap [k^2]| \ge k$. Hence there is a $T \in \binom{[k^2]}{k}$ such that $T \subseteq P$. Consequently, $c(P,S_T) \ge |P \cap T| = k$. This shows that the integrality gap of this instance is at least $k$.

It remains to show that $k = \Omega(\log K / \log\log K)$ for the number $K = |\binom{[k^2]}{k}|$ of scenarios. By Stirling's formula, we compute $\log(K) =  \Theta(\log\binom{k^2}{k}) = \Theta(k \log k)$. Consequently, $\log K / \log\log K = \Theta(k \log k / \log(k \log k)) = \Theta(k)$.  
 
\bibliographystyle{alpha}
\bibliography{references}

\end{document}